\documentclass[conference]{IEEEtran}
\IEEEoverridecommandlockouts
\usepackage{cite}
\usepackage{amsmath,amssymb,amsfonts,amsthm}
\usepackage{graphicx}

\usepackage[T1]{fontenc}

\usepackage{xcolor}
\usepackage{stackengine}
\usepackage{verbatim}

\usepackage{cite}

\usepackage{epstopdf}
\usepackage{enumerate}
\usepackage{subfigure}
\usepackage{color}

\usepackage{url}

\usepackage{balance}

\usepackage[caption=false,font=footnotesize]{subfig}

\usepackage{floatrow}

\newtheorem{theo}{Theorem}

\def\delequal{\mathrel{\ensurestackMath{\stackon[1pt]{=}{\scriptstyle\Delta}}}}
\def\BibTeX{{\rm B\kern-.05em{\sc i\kern-.025em b}\kern-.08em
		T\kern-.1667em\lower.7ex\hbox{E}\kern-.125emX}}

\def\x{{\mathbf x}}

\def\t {{\mathbf \Theta}}
\def\tk {{\mathbf \Theta_{k}}}
\def\alp {{\boldsymbol \alpha}}
\def\pvk {{\mathbf p(\nu_k)}}
\def\pv {{\mathbf p(\nu)}}
\def\Pv {{\mathbf P(\nu)}}

\def\h{{\mathbf h}}
\def\x{{\mathbf x}}

\def\c{{\mathbf c}}
\def\n{{\mathbf n}}
\def\y{{\mathbf y}}
\def\g{{\mathbf g}}
\def\th {{\mathbf \Theta}}
\def\A {{\mathbf A}}
\def\a{{\mathbf a}}
\def\b{{\mathbf b}}
\def\R{{\mathbf R}}

\def\E{{\mathbf E}}

\def\C{{\mathbf C}}
\def\c{{\mathbf c}}
\def\n{{\mathbf n}}
\def\p {{\mathbf p}}

\def\mbJ{{\mathbf J}}
\def\mbI{{\mathbf I}}
\def\e {{\mathbf e}}

\def\A{{\mathbf A}}
\def\a{{\mathbf a}}
\def\b{{\mathbf b}}

\def\bzero{\boldsymbol{0}}

\def\Omegaa{{\mathbf \Omega}}
\def\Psii{{\mathbf \Psi}}
\def\bmu{{\mathbf \mu}}
\def\LoS{{\textrm{LoS}}}
\def\NLoS{{\textrm{NLoS}}}
\def\PRI{{\textrm{PRI}}}
\def\bJ{{\mathbf J}}

\def\Re#1{\mathrm{Re}\left(#1\right)}
\def\Im#1{\mathrm{Im}\left(#1\right)}

\begin{document}
	\title{IRS-Aided Radar: Enhanced Target Parameter  Estimation via Intelligent Reflecting Surfaces
		\thanks{This work was supported in part by the U.S. National Science Foundation Grant ECCS-1809225.}
	}\author{\IEEEauthorblockN{Zahra Esmaeilbeig$^{\dag,1}$, Kumar Vijay Mishra$^{\ddag,2}$, and Mojtaba Soltanalian$^{\dag,3}$}
		\IEEEauthorblockA{$^{\dag}$ Department of Electrical and Computer Engineering\\
			University of Illinois at Chicago, Chicago, Illinois, USA\\
			\IEEEauthorblockA{$^{\ddag}$United States DEVCOM
				Army Research Laboratory, Adelphi, Maryland, USA\\
				Email: \{$^{1}$zesmae2, $^{3}$msol\}@uic.edu,$^{2}$kvm@ieee.org}
		}
	}
	\maketitle
\begin{abstract}
	The intelligent reflecting surface (IRS) technology has recently attracted a lot of interest in wireless communications research. An IRS consists of passive reflective elements capable of tuning the phase, amplitude, frequency, and polarization of impinging waveforms.  We investigate the deployment of IRS to aid radar systems when the line-of-sight (LoS) link to the targets is weak or blocked. We demonstrate that deployment of multiple IRS platforms provides a virtual or non-line-of-sight (NLoS) link between the radar and target leading to an enhanced radar performance. Numerical experiments indicate that the IRS enhances the target parameter estimation when the LoS link is weaker by \textasciitilde$10^{-1}$ in comparison to the NLoS link.
\end{abstract}
\section{Introduction}
\label{sec:intro}
An intelligent  reflecting surface (IRS) is composed of a large number of passive reconfigurable meta-material elements, which reflect the incoming signal by introducing a predetermined phase shift \cite{hodge2020intelligent}. In a communication system, this phase shift is controlled via an external signal transmitted by the base station (BS) through a backhaul control link. As a result, the incoming signal from the BS is manipulated in real time, thereby, efficiently reflecting the received signal toward the users~\cite{ozdogan2020using,asadchy2016perfect,estakhri2016wave,elbir2020survey}. The IRS  technology  has appeared in wireless communications, also under other names including  large intelligent  surface and  software-controlled metasurfaces~\cite{nadeem2019large,bjornson2019massive,liaskos2018new}. 
Several promising IRS use-cases including range extension to users with obstructed  direct links~\cite{nadeem2019large}, physical layer security~\cite{guan2020intelligent}, and unmanned air vehicle (UAV) communications~\cite{li2020reconfigurable} have been studied. Some prior works on IRS-assisted signal transmission are~\cite{wu2019intelligent,li2020reconfigurable,wang2019joint,huang2018achievable,wu2018intelligent}. 

In this context, IRS deployment has an untapped potential in radar system design and signal processing for target detection and estimation~\cite{khobahi2020deep}-~\cite{bose2021efficient}. In an IRS-aided radar, the surface manipulates the signal coming from the radar transmitter (target) and reflects it toward the target (radar receiver) (Fig.~\ref{fig::3}).  Lately, IRS has emerged as a promising and  cost-effective solution to establish robust connections even when  the line-of-sight (LoS) link is  blocked by obstructions~\cite{tan2018enabling}. There have been several prior works on non-line-of-sight (NLoS) %\textcolor{red}{use only NLoS, not NLOS. Check everywhere.} 
radar systems without the aid of IRS \cite{watson2019non,guo2020behind,solomitckii2021radar}. However, these techniques require knowledge of the entire geometric structure of the environment. In addition, processing the mutipath returns from a target is computationally demanding. The IRS-aided NLoS radar is a paradigm shift because the location of the IRS platforms and flexibility in beamforming via IRS are sufficient to perform target detection and estimation. By smartly tuning the phase shifts of IRS passive elements, effective NLoS or virtual LoS links are created thereby yielding a more reliable sensing of targets. 

The IRS-aided radar for NLoS scenarios was introduced in \cite{aubry2021ris} and extended to multiple-input multiple-output (MIMO) radar in \cite{buzzi2021radar,zhang2022}. In this paper, we develop a mathematical model for  IRS-aided radar parameter estimation and investigate  the potential gains associated with the IRS deployment in such settings. Contrary to most prior works \cite{aubry2021ris,buzzi2021radar,zhang2022} %\textcolor{red}{no need for - here. Just put commas. LaTeX will format it then.} 
that focused on the NLoS sensing via a single IRS, we incorporate  multiple IRS platforms~\cite{vijay2022,elbir2022rise,wei2022}. We develop the general signal model for a  multiple IRS-aided radar, in which the IRS acts as a phase shift component and benchmark the performance of IRS platforms through mean square error of  target parameter estimation. % We  then consider the single-input single-output (SISO) radar 
%\textcolor{red}{makes no sense why we would describe MIMO model and then shift to SISO?} 
We derive the best linear unbiased estimator (BLUE) for estimating the target back-scattering coefficient. Our numerical experiments show that using IRS even with  randomly chosen phase shifts improve the mean-squared-error of target parameter estimation. We further study the optimization of the IRS platform  by designing phase shifts to specifically minimize the mean-squared-error of target parameter estimation. As expected, the optimized IRS case leads to lower estimation error in comparison with the non-optimized IRS. We further derive the Cram\'er-Rao bound (CRB) for estimation of the target parameter and illustrate it  for the LoS and NLoS scenarios. 

Throughout this paper, we use bold lowercase letters for vectors and bold uppercase letters for matrices. The notations $(\cdot)^T$ and $(\cdot)^H$ denote the vector/matrix transpose and the  Hermitian transpose, respectively. The symbols $\odot$ and $\otimes$ stand for the Hadamard (element-wise)  and Kronecker product of matrices; $\textrm{Tr}(\cdot)$ is the trace operator for matrices;  $\textrm{Diag}(.)$ denotes the diagonalization operator that produces a diagonal matrix  with same diagonal entries as the entries of its  vector argument; and $\textrm{diag}(.)$ outputs a vector containing the  diagonal entries of the input matrix. $\|\cdot\|_2$ is the $\ell_2$ norm. Finally, $\textrm{arg}(.)$, $\operatorname{Re}(\cdot)$ and $\operatorname{Im}(\cdot)$ return the arguments, real part and imaginary part  of a complex input vector, respectively.

\section{System Model}\label{sec::model}
 %--------------------------------------------------
	\begin{figure}[t]
		\centering
		\includegraphics[width=0.75\columnwidth]{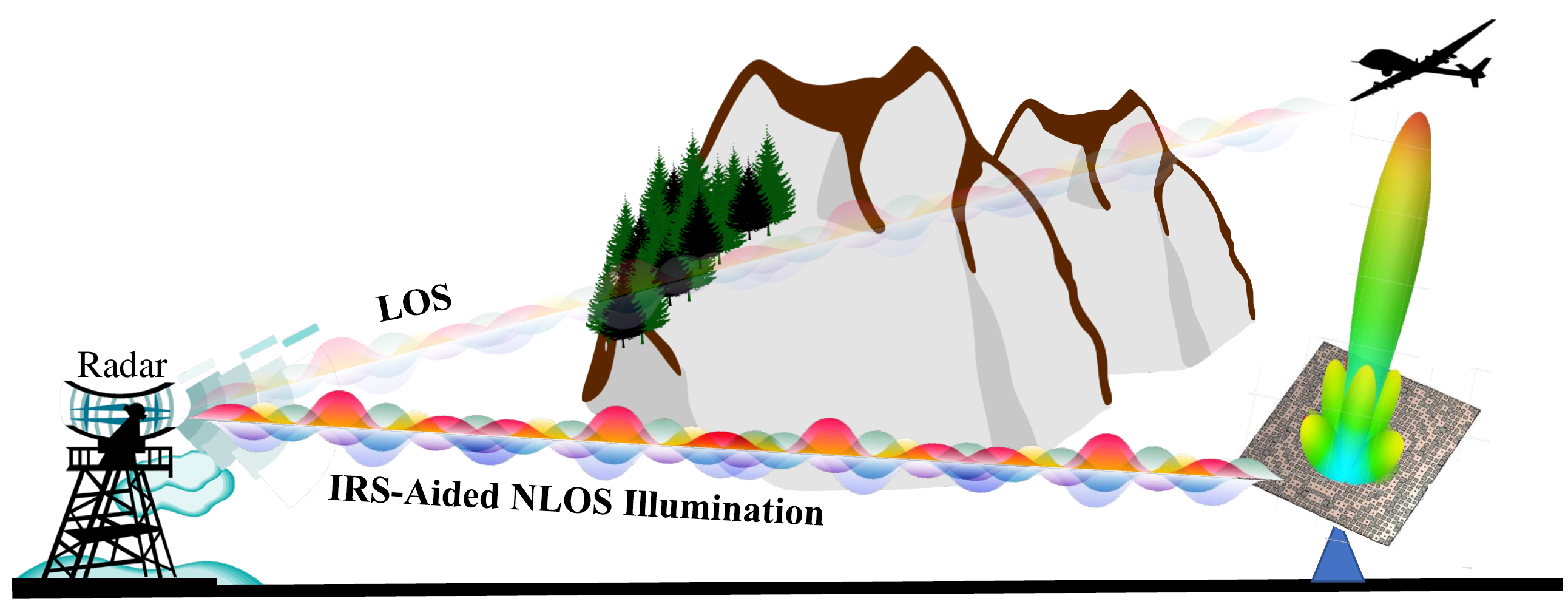}
		\caption{An illustration of IRS-aided radar operation. The IRS  creates effective virtual LoS links between the radar and the desired targets. Large IRS platforms may be required for far-field deployment.}
		\label{fig::3}
	\end{figure}
	%--------------------------------------------------
Consider a pulse-Doppler radar system that transmits a train of $N$ uniformly-spaced pulses $x(t)$, each of which is nonzero over the support $[0, T_p]$, as 
\begin{equation}
s(t)=\sum_{n=0}^{N-1}x(t-nT_{_{\PRI}}), \;   0\leq t \leq (N-1)T_{_{\PRI}},
\end{equation}
where $T_{_{\PRI}}$ is the pulse repetition interval (PRI). The entire duration of all $N$ pulses is the coherent processing interval (CPI) (``slow time''). 
%\textcolor{red}{Now write, consider a target  with reflectivity .... Cite only skolnik's book \cite{skolnik2008radar} as reference for the target model, See how to write this para using the first para of Section 1.3.1 of my chapter \cite{mishra2019sub} \url{https://arxiv.org/pdf/1803.01819.pdf} }
Assume the target scene consists of a single Swerling-0 model \cite{skolnik2008radar} moving  target %\textcolor{blue}{located at range-time delay $\tau$}
, characterized by unknown parameters: complex reflectivity/amplitude $\alpha_{_T}$ and the normalized target Doppler shift $\nu$ (expressed in radians) relative to the radar; note that, in the formulation of this paper, the target range is assumed to be known. 

In the absence of an IRS, the transmit signal is reflected back from the target and collected by the radar receiver (Fig.~\ref{fig::LOS-fig}). The baseband continuous-time received signal %for a single pulse 
is\par\noindent\small % \textcolor{red}{where are pulses and PRI in the equation below?}
\begin{align}
 y(t)&=\alpha_{_T}h_{_{\LoS}} \sum_{n=0}^{N-1}x(t-nT_{_{\PRI}}) e^{j \nu t} +n(t),\nonumber\\
 &\approx \alpha_{_T}h_{_{\LoS}} \sum_{n=0}^{N-1}x(t-nT_{_{\PRI}}) e^{j \nu nT_p} +n(t),
\end{align}\normalsize
where  $n(t)$ is random additive noise, %\textcolor{red}{this is not a `sample'. just noise. sample is when you have performed sampling, i.e., in discrete-time domain} 
%~\cite{gini2012waveform,aubry2012cognitive,ameri2019one,aubry2013knowledge}.  
$h_{_{\LoS}}$ accounts for  the radar-target-radar channel state information (CSI), and the last approximation follows from the fact that $\nu \ll 1/T_p$ so that the phase rotation within the CPI could be approximated as a constant. 

Each snapshot of the received signal are sampled at the rate $1/T_{_{p}}$ yielding a total of $\lfloor T_{\textrm{PRI}}/T_p\rfloor$ ``fast-time'' samples. As mentioned earlier, we assume the range of the target is known.
%from the range bin corresponding to the delay $\tau$.
%\textcolor{red}{what is Tc? pulsewidth?} %\textcolor{red}{no need of `let us' in paper. Directly write `Define'} 
 At this fixed target range in fast-time, we collect all $N$ slow-time samples of the received signal corresponding to each pulse in the vector $\y=[y(0),y(T_{_{p}}),\ldots,y((N-1)T_{_{p}})]^T$ as %\textcolor{red}{you earlier wrote the signal in `matrix form'. But the following is a vector form. I have removed `matrix form' phrase now}
\begin{equation}\label{eq::2}
\y=\alpha_{_T}h_{_{\LoS}}\left[  \x  \odot\pv \right]  +\n,
\end{equation}
where $\p(\nu)=[1,e^{jT_{_{p}}\nu} ,\ldots,e^{j(N-1)T_{_{p}}\nu }]^T$.
$\x = [\x(0),\x(T_{_{p}}),\ldots,\x((N-1)T_{_{p}})]^T$ and $\n =[n(0),n(T_{_{p}}),\ldots,n((N-1)T_{_{p}})]^T$ are, respectively, transmit and noise signal vectors and $\n$ is zero-mean random vector with covariance matrix $\R$~\cite{elbir2019cognitive}. 

We now consider the received signal in the presence of an IRS (Fig.~\ref{fig::4}). Assume $K$ IRS platforms are deployed and NLoS paths are realized through IRS platforms between the radar and target. An IRS is typically  deployed as an array of discrete scattering elements. Each element (also known as a meta-atom or lattice) has the ability to introduce a phase shift to an incident wave. Assume that each IRS is equipped  with $M$ reflecting elements. Each IRS element reflects the incident signal with a phase shift and amplitude change that is configured via a smart controller. Define
	\begin{equation}
		\t_k=\text{\textrm{Diag}}(\beta_{k,1}e^{\textrm{j}\theta_{k,1}},\ldots,\beta_{k,M}e^{\textrm{j}\theta_{k,M}})
		\label{eq::4}
	\end{equation}
	as the phase-shift matrix of the $k$-th IRS, where $\theta_{k,m}\in[0,2\pi]$, and $\beta_{k,m}\in [0,1]$ are, respectively, phase shift and amplitude reflection gain associated with the $m$-th passive element of the $k$-th IRS. In general, it suffices to design only the phase-shift so that  $\beta_{k,m}=1$ for all $(k,m)$~\cite{vijay2022,ahmed2022joint,gini2012waveform,aubry2012cognitive}.

The path via the $k$-th IRS is characterized by the corresponding CSI  $h_{_{\NLoS,k}}$, IRS phase shift matrix $\tk$, Doppler shift $\nu_k$, and the target reflectivity amplitude $\alpha_{_{T,k}}$. Denote %\textcolor{red}{this is not an assumption. Use language like `Denote' ..... by ...symbol} 
the radar-IRS$_k$ and target-IRS$_k$ CSI by, respectively, $\g_k \in \mathbb{C}^{M}$ and $\h_k \in \mathbb{C}^{M}$. The CSI for all the paths between radar, target and IRS platforms are assumed to be well-estimated through suitable channel estimation techniques~\cite{zheng2022survey}. We define %\textcolor{red}{do not use `let' and `suppose' in research papers, Use denote, define, etc. } 
$\h_{_\NLoS}=[h_{_{\NLoS,k}},\ldots,h_{_{\NLoS,K}}]^T$  as  the NLoS CSI  vector. Using the channel reciprocity of IRS \cite{tang2021channel}, the NLoS CSI is
	\begin{equation}\label{eq::CSI-vec}
	\h_{_\NLoS}=\left[\left|\h_1^H\th_1\g_1\right|^2,\ldots,\left|\h_K^H\th_K \g_K\right|^2\right]^T.
	\end{equation}
The received signal is a superposition of the reflected signals from all NLoS paths as %yields the received signal in a multiple IRS-aided radar as
\begin{equation}\label{eq::NLoS-SISO}
\y=\alpha_{_T}h_{_{\LoS}}\left[  \x  \odot\pv \right]+\sum_{k=1}^{K}\alpha_{_{T,k}}h_{_{\NLoS,k}}\left[ \x  \odot\pvk \right] +\n.
\end{equation}

 While both LoS and NLoS signals are available at the receiver, we aim to show the effectiveness of IRS-created NLoS in overcoming obstructed or weak LoS links. Therefore, throughout this paper, we consider the case when the LoS link strength is  insignificant, i.e. $\h_{\LoS}\simeq \bzero$ and the signal received through NLoS is used to obtain target information. Denote the complex reflectivity vector by %$\alpha_{_{T,k}}$ 
	 $\alp=[\alpha_{_{T,1}},\alpha_{_{T,2}},\ldots,\alpha_{_{T,K}}]^T$
	 of a moving target for $k \in \{1,\ldots,K\}$ paths. % as \textcolor{red}{where is the design?} a mismatched filter \cite{stoica2011optimization} for estimating target back-scattering coefficients
	%$\alp=[\alpha_{_{T,1}},\alpha_{_{T,2}},\ldots,\alpha_{_{T,K}}]^T$. Note that 
	Rewrite the received signal in~\eqref{eq::NLoS-SISO} compactly as %may be written as 
	\begin{equation}\label{eq::6}	
		\y=\A \alp +\n,
	\end{equation}
	where %, for the case of multiple IRS platforms, $K\geq1$,  
	$\A=[\a_1,\ldots,\a_K]\in \mathbb{C}^{N \times K}$ is the sensing matrix with columns
	\begin{equation}
	\label{eq::7}	
		\a_k \delequal 	\h_{_{\NLoS,k}} \left[ \x \odot \p(\nu_k) \right].
	\end{equation}
	%--------------------------------------
	\begin{figure}[t]
    \centering
	\includegraphics[width=0.68\columnwidth]{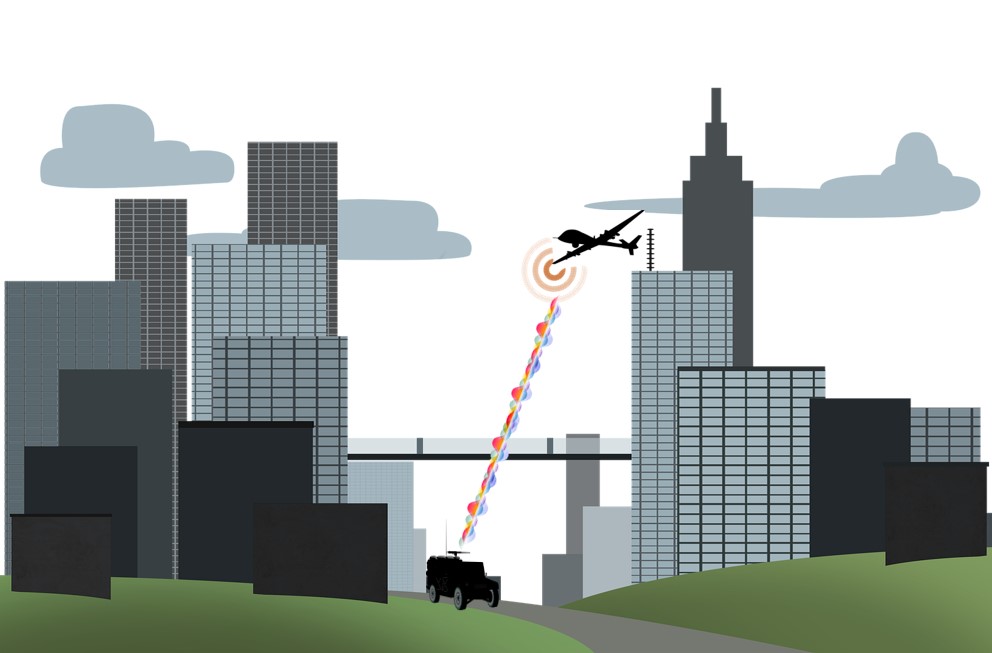}
		\caption{LoS  link between the  radar and the target }
		\label{fig::LOS-fig}
	\end{figure}
	%--------------------------------------
	 Our goal is to obtain BLUE for $\alpha_{_{T,k}}$ for all $k$ paths.
	 %In the following, derive the  BLUE estimator, as well as the CRB as a lower bound on the estimation performance of $\alpha_{_{T,k}}$.
	
	\nocite{esmaeilbeig2020deep}
	\nocite{eamaz2021modified}
	\section{IRS-Aided Target Parameter Estimation}
	\label{sec::prob}
	Following the Gauss-Markov theorem~\cite{kay1993fundamentals}, the BLUE for $\alp$ is
	\begin{equation}
	\label{eq::8}
	\hat{\alp}=\left(\A^H\R^{-1}\A\right)^{-1}\A^H \R^{-1}\y.
    \end{equation}
	
    The covariance  matrix of $\hat{\alp}$ is  
	\begin{equation}
	\label{eq::10}	
		\C_{\hat{\alp}}=\left(\A^H\R^{-1}\A\right)^{-1},
	\end{equation}
	with the minimum achieved variance for $\hat{\alpha}_k$ given by
	$\textrm{var}(\hat{\alpha_k})=\left[\left(\A^H\R^{-1}\A\right)^{-1}\right]_{kk}$. The  overall mean-squared-error (MSE) of the proposed estimator is thus  $\text{MSE}(\hat{\alp})= \textrm{Tr} \left(\left(\A^H\R^{-1}\A\right)^{-1}\right)$. The optimal  phase-shift matrices $\{\tk\}, ~k\in\{1,\ldots,K\}$ to minimize the MSE of target  cross-section parameter are obtained by solving the following problem\par\noindent\small
	\begin{equation}\label{eq::13}	
		\underset{\tk,~k\in\{1,\ldots,K\}}{\textrm{minimize}} \text{MSE}(\hat{\alp})=\underset{\tk,~k\in\{1,\ldots,K\}}{\textrm{minimize}}\textrm{Tr}\left(\left(\A^H\R^{-1}\A\right)^{-1}\right). 
	\end{equation}\normalsize
%\textcolor{red}{Why? This is very confusing. If you are not doing MIMO, why introduce it in the first place?}	
The following theorem states that the  optimal phase shifts of the $K$ different IRS platforms are decoupled. Also, the optimal phase shift for IRS$_k$ %\textcolor{red}{write this as `$k$-th IRS'. Always use -th e.g. $k$-th IRS, $p$-th vector, $n$-th sample. Not $n^{\textrm{th}}$ sample.}
compensates for the total phases in the channels $\g_k$ and $\h_k$.
\begin{theo}
The solution to the optimization problem\par\noindent\small
		\begin{equation}
			\underset{\tk,~~k\in\{1,\ldots,K\}}{\textrm{minimize}}~ \textrm{MSE}(\hat{\alp}),
			\label{eq::17}	
		\end{equation}
		is 
		\begin{equation}\label{eq::18}
			\t_k^*=\textrm{Diag}\left(e^{\textrm{j}\text{\textrm{arg}}(\c_k)}\right),
		\end{equation}\normalsize
		where $\c_k\delequal\text{\textrm{Diag}}(\g_k)^H \h_k$.
	\end{theo}
\begin{proof}
Define the IRS-observed Doppler shift matrix as $\Pv \delequal [\p(\nu_1),\ldots,\p(\nu_K)]\in \mathbb{C}^{N \times K}$. Incorporating this definition in~\eqref{eq::6}-\eqref{eq::7}, we have $\A=\textrm{Diag}(\x) \Pv\textrm{Diag}(\h_{_\NLoS})$, which is used to compute\par\noindent\small
		\begin{equation}
				\left(\A^H \R^{-1} \A\right)^{-1}= \A^{-1}\R \A^{-H}
				=\textrm{Diag}(\h_{_\NLoS})^{-1} \, \Psi \, \textrm{Diag}(\h_{_\NLoS})^{-H},\\
			\label{eq::22}
		\end{equation}\normalsize
		with $\Psii\delequal \Omegaa^{-1} \R  \Omegaa ^H$ and $\Omegaa\delequal \textrm{Diag}(\x)\Pv$. Substituting~\eqref{eq::22} in~\eqref{eq::13}  yields the optimization problem\par\noindent\small
		\begin{equation}
		\underset{\tk,~k\in\{1,\ldots,K\}}{\textrm{minimize}}\; \sum_{k=1}^{K} |h_{_\NLoS,k}|^{-2}\Psi_{kk}=\underset{\tk,~k\in\{1,\ldots,K\}}{\textrm{maximize}} \;\left|\h_k^H \t_k \g_k\right|. \nonumber
			\label{eq::23}	
		\end{equation}\normalsize
Considering the property  $\a \odot \b= \textrm{Diag}(\a)\b$ of the Hadamard product as well as the  diagonal structure of  $\t_k$, we write
\begin{equation}		\label{eq::15}
\h_k^H \t_k \g_k=\h_k^H [\textrm{diag}(\t_k)\odot \g_k]= \c_k^H \textrm{diag}(\t_k).
\end{equation}
It now follows that the optimal solution to~\eqref{eq::13} is~\eqref{eq::18}.
\end{proof}
\begin{figure}[t]
	\centering
	\vspace{0.05in}
    \includegraphics[width=0.68\columnwidth]{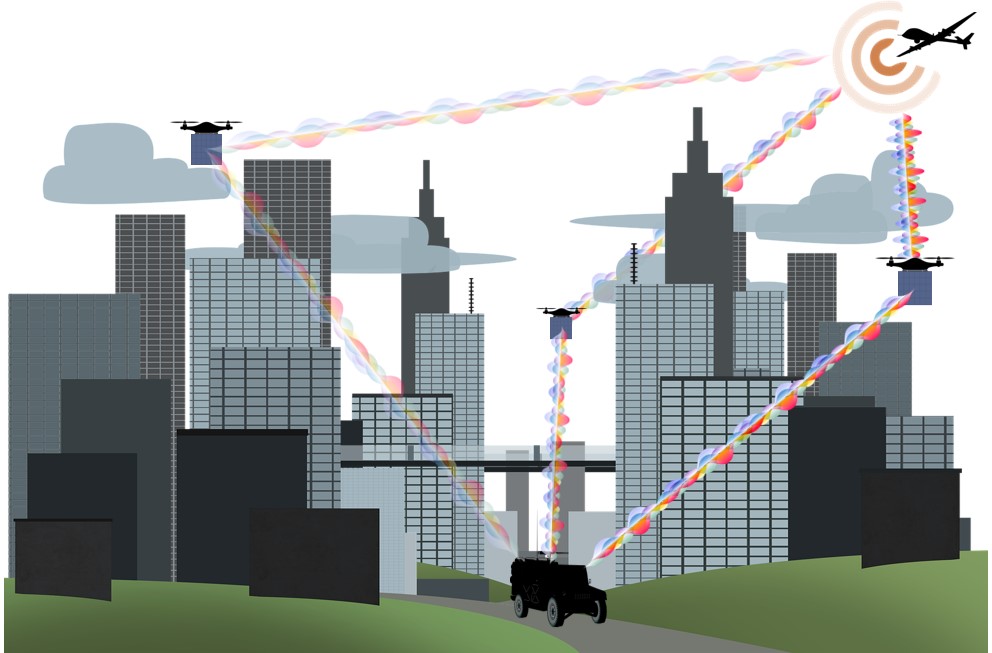}
	\caption{NLoS or virtual LoS link between the  radar and the  target provided by  $K=3$ IRS platforms.}
	\label{fig::4}
\end{figure}
\section{Error Bound Analysis}
 We analyze the CRB of the proposed IRS-aided target parameter estimation. Assume
 \begin{equation}
     \Tilde{\alp}=[\alp_{R}^T,\alp_I^T]^T
\end{equation}
with $\alp_{R}=\Re{\alp}$, $\alp_I=\Im{\alp}$. For an unbiased estimator of the parameter $\alp$, the covariance matrix of  $\hat{\tilde{\alp}}$ is lower bounded as $\C_{\hat{\Tilde{\alp}}}=\E\{{(\hat{\Tilde{\alp}}-\Tilde{\alp})(\hat{\Tilde{\alp}}-\Tilde{\alp})^H}\}\geq \C_{_{\textrm{CRB}}}$, in the sense that the difference  $\C_{\hat{\Tilde{\alp}}}-\C_{_{\textrm{CRB}}}$ is a positive semidefinite matrix~\cite{lu2021intelligent,peebles2007radar,kay1993fundamentals,van2004optimum}.
 %The first step to  derive the $C_{_{\textrm{CRB}}}$ is to  compute the  Fisher information matrix (FIM) $\mbJ$. 
% \textcolor{red}{are you sure? complex-valued CRB would have different expressions. Why not just do the analysis of complex CRB? See the derivation of complex CRB in S M Kay's book. I remember it is in the Appendix of one of the chapters}
 We divide the Fisher information matrix (FIM), $\bJ$ into submatrices as
\par\noindent\small
\begin{equation}
\label{eq::fisher-submatrices_alpha}
		\bJ=\begin{bmatrix}
			\bJ_{\alp_R,\alp_R}& \bJ_{\alp_R,\alp_I}\\
			\bJ_{\alp_I,\alp_R}&\bJ_{\alp_I,\alp_I}
		\end{bmatrix}.
\end{equation}\normalsize
 Using the Slepian-Bangs formula~\cite{stoica2005spectral} for the  observation vector $\y$,  with a Gaussian  distribution $\y \sim N(\bmu,\R)$, the  $(m,n)$-th element of the Fisher information matrix (FIM) is\par\noindent\small
\begin{equation}
\bJ_{mn}=\textrm{Tr} \left(\R^{-1}\frac{\partial \R }{\partial\Tilde{\alp}_m} \R^{-1}\frac{\partial \R }{\partial\Tilde{\alp}_n}\right)+2\operatorname{Re}\left(\frac{\partial\bmu}{\partial\Tilde{\alp}_m}^H\R^{-1}\frac{\partial\bmu}{\partial\Tilde{\alp}_n}\right).
\end{equation}  \normalsize
Following~\eqref{eq::6}, we have $\bmu = \A \alp$. The FIM elements are
\par\noindent\small
\begin{align}\label{eq::25}
[\bJ_{\alp_R,\alp_R}]_{mn}&=2\operatorname{Re}\left(\frac{\partial\bmu}{\partial\alp_{R_m}}^H\R^{-1}\frac{\partial\bmu}{\partial\alp_{R_n}}\right)=2\operatorname{Re}\left(\a_m^H \R^{-1} \a_n\right)\nonumber\\
&=2\operatorname{Re}\left(\a_m^H\R^{-1}\a_m\right)=2\operatorname{Re}\left(\e_m^T\A^H\R^{-1}\A\e_n\right),
\end{align}\normalsize
where $\e_m$ is a $K \times 1$ vector, whose $m$-th element is unity and remaining elements are zero. Similarly, other submatrices of  the FIM are \par\noindent\small
\begin{align}\label{eq::J-alpha}
\bJ_{\alp_{I},\alp_{I}}&=\bJ_{\alp_{R},\alp_{R}}=2\Re{\A^H \R^{-1} \A},\nonumber\\
\bJ_{\alp_{R},\alp_{I}}&=-\bJ_{\alp_{I},\alp_{R}}=-2\Im{\A^H \R^{-1} \A}.
\end{align}
\normalsize
substituting~\eqref{eq::J-alpha} in~\eqref{eq::fisher-submatrices_alpha}, we get
\begin{equation}
\bJ=2\Re{[1 \; j]^H\otimes[1 \; j]\otimes (\A^H \R^{-1} \A)}, 
\end{equation}
the inverse of which yields %is  the  $C_{_{\textrm{CRB}}}$, i.e.
$\C_{_{\textrm{CRB}}}=\mbJ^{-1}$.
\begin{comment}
For a large  FIM, taking the inverse analytically  is usually done using the blocked matrix inversion theorem~\cite{van2004optimum}. Assume  $K=2$ as an  example, we  have  the  CRB on the estimator  $\hat{\alp}_1$ as $var(\hat{\alp_1})=[\C_{\hat{\alp}}]_{11}\geq[J^{-1}]_{11}= (J_{11}-J_{12}J_{22}^{-1} J_{21})^{-1}$, which using~\eqref{eq::25} becomes
\begin{equation}
\small
var(\hat{\alp_1})\geq \frac{\sigma_n^2}{2}\left(||\a_1||^2-\frac{\operatorname{Re}(\a_1^Ha_2)^2}{||\a_2||^2}\right)^{-1}.
\normalsize
\end{equation}
\end{comment}
\section{Numerical Experiments}
\label{sec::num}
	We validated the performance of target parameter estimation in IRS-aided radar through numerical experiments. Throughout our experiments, $\x \in \mathbb{C}^{N}$ is a unimodular code that is randomly chosen with the length $N=50$ i.e., $\x_n=e^{j\phi_n}, \quad n \in\{1,\ldots, 50\}$~\cite{soltanalian2014designing}. We set $K=5$ and $M=10$. We generated  the noise vector $\n$ from  an independent and identically distributed random Gaussian process i.e. $\R=\sigma_n^2 \mbI$.
  %It is notable that proper selection of these parameters is also an open %IRS design problem for future studies. 
  We consider the following scenarios:
	\begin{itemize}
		\item An LoS path (Fig.~\ref{fig::LOS-fig}) is present between the radar and target with the  CSI $\h_{_\LoS}$.
		%such that $||\r_{_{\textbf{LOS}}}||_{_{2}}=||\H_{rtr}||_{_{2}}$ . 
		\item There is an NLoS  path (Fig.~\ref{fig::4}) through $K$ IRS platforms with non optimal $\theta_{k,m}$, $k\in\{1,\ldots,K\}$ and  $m\in\{1,\ldots,M\}$ chosen randomly in the interval $[0,2\pi)$. 
		\item There is an NLoS path with the  CSI  $\h_{_{\textbf{NLoS}}}$. The IRS phase-shift parameters $\t_k, k\in\{1,\ldots,K\}$ are optimized and set according to~\eqref{eq::18}. 
	\end{itemize}
	We define the \emph{LoS-to-NLoS signal-to-noise ratio (SNR)} as \par\noindent\small
	\begin{equation}
	\gamma\delequal\frac{|\alpha_{_T} h_{\LoS}|^2}{||\alp^T \h_{_\NLoS}||^2_{_2}},
	\end{equation}\normalsize
    which governs the relative strengths of the LoS and NLoS links. 
    %----------------------------------
	\begin{figure}[t]
		\includegraphics[width=0.87\columnwidth]{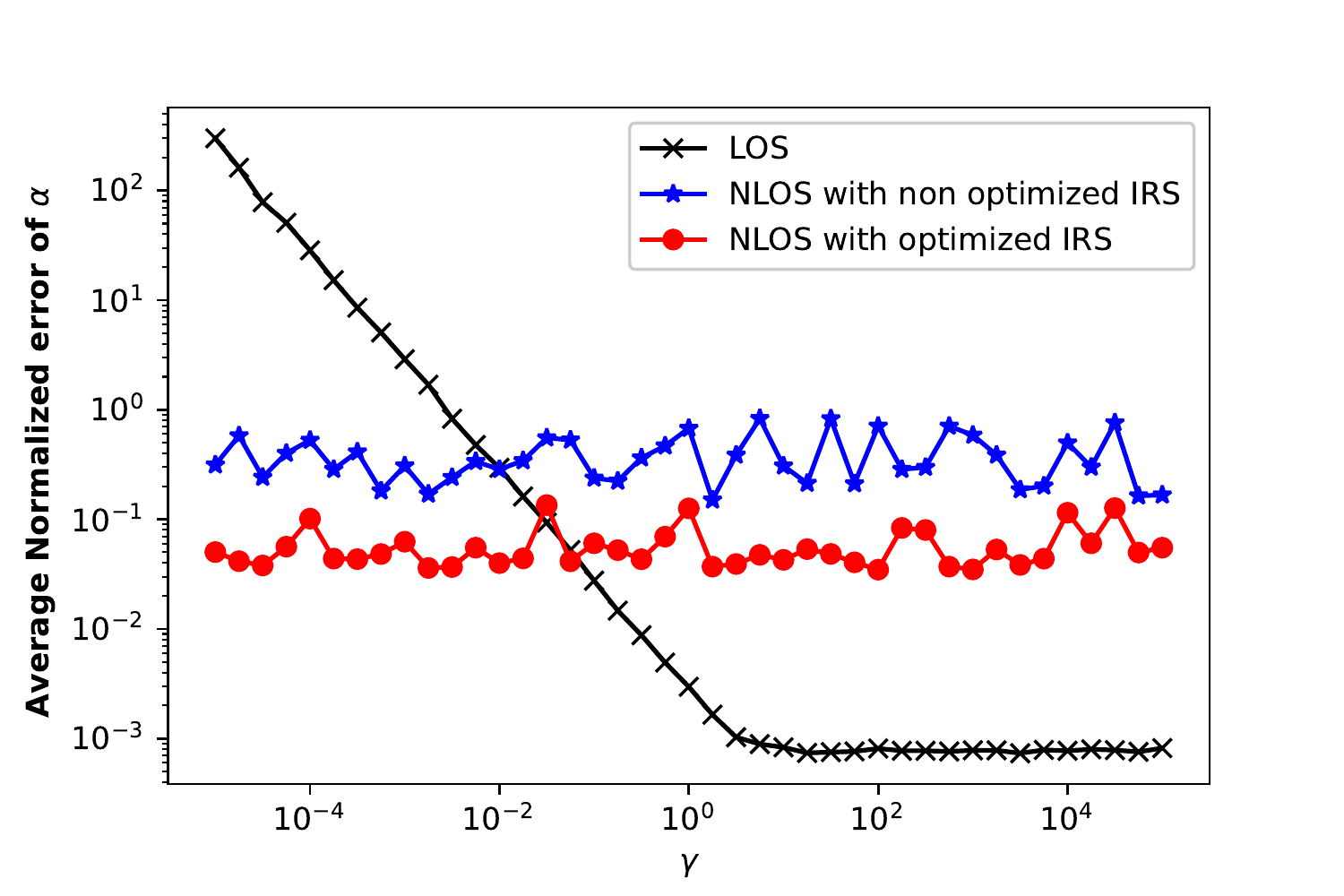}
		\caption{NMSE for estimation  of target scattering coefficient $\alp$ for different values of  $\gamma\in [10^{-5},10^{5}]$, $K=5$ IRS platforms, and $M=10$.}
		\label{fig::2}
	\end{figure}
	%----------------------------------
	%----------------------------------
    \begin{figure}[t]
		\includegraphics[width=0.87\columnwidth,draft=false]{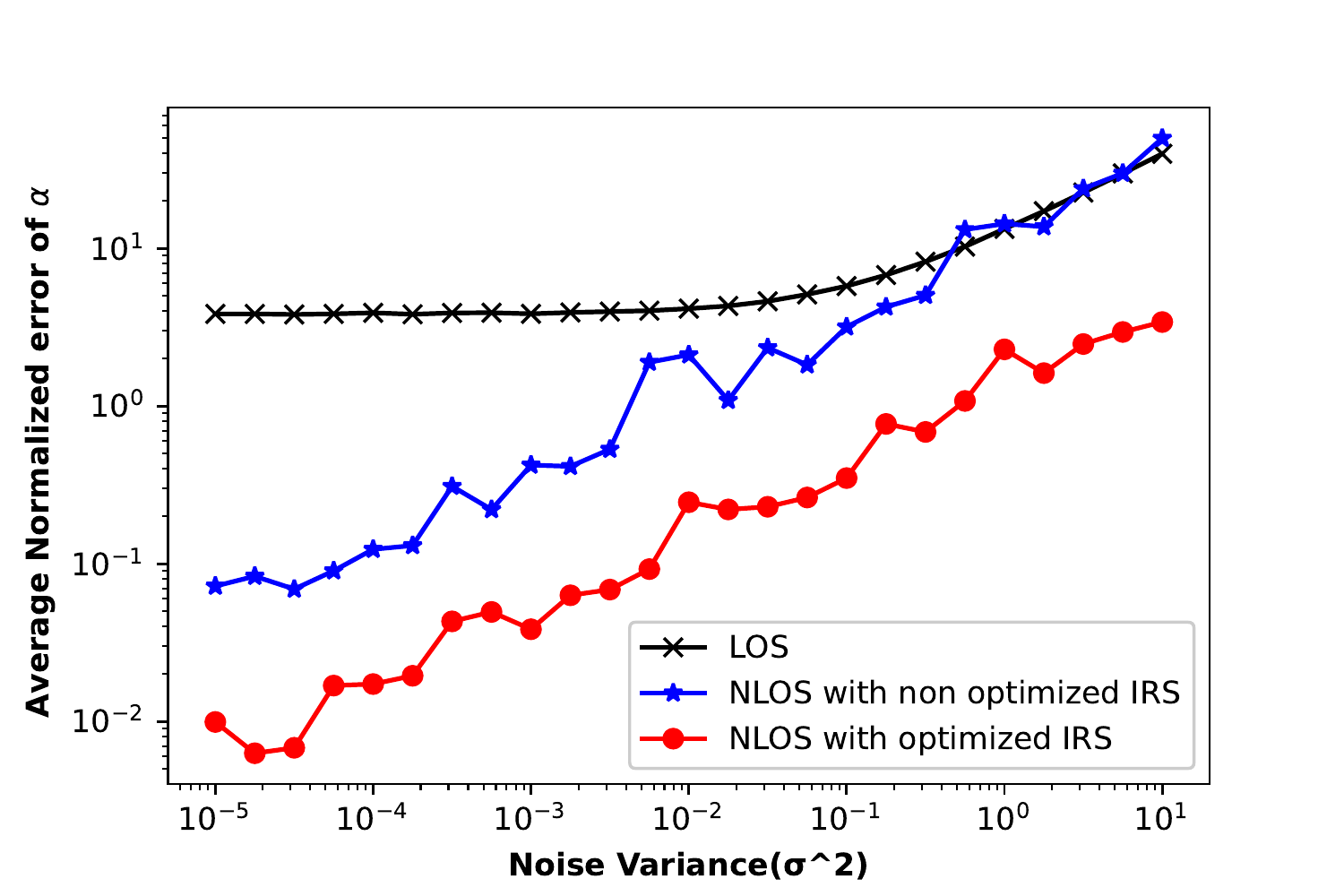}
		\caption{Average NMSE of the recovered target scattering coefficient $\alp$ for different noise variances $\sigma_n^2 \in \left[10^{-5},1\right]$, with $K=5$  IRS platforms, $M=10$ reflecting elements and the LoS-to-NLoS SNR $\gamma=10^{-2}$.}
		\label{fig::6}
	\end{figure}
	%----------------------------------
	We use the normalized estimation error of the
	back-scattering coefficient $\alp$, defined by 
	$\text{NMSE}\delequal\frac{||\alp-\hat{\alp}||_{_2}}{||\alp||_{_2}}$,
	as a measure of performance for our estimators. In Fig.~\ref{fig::2}, we illustrate the effectiveness of the optimized and non-optimized IRS over different strengths of the link between the radar and the target. In order to control the LoS-to-NLoS SNR $\gamma$, we generated the LoS and NLoS channels such that $|\alpha_{_T} h_{\LoS}|^2=\gamma$ and $||\alp^T \h_{_\NLoS}||^2_{_2}=1$. The CSI for all channels involved is sampled from an independent circularly symmetric complex Gaussian random vector with zero mean and variance of unity and scaled  such that we have $\gamma\in [10^{-5},10^{5}]$~\cite{vijay2022}. The Doppler shifts in  both the LoS and NLoS  is chosen from a random uniform distribution on $ [-0.5,0.5)$~\cite{ameri2019one}. The results are averaged over $10^3$  Monte-Carlo trials. The perturbations in Figs.\ref{fig::2}-\ref{fig::9} arise from the randomness of the channels and Doppler shifts in each Monte-Carlo sample.
	
	Fig.~\ref{fig::2} indicates that the IRS overcomes the LoS links as weak as $~10^{-1}$ times the NLoS link.
	%In other words, the optimized and non-optimized  IRS will remain effective with increasing LOS link strength until we approach $\gamma=10^{-1}$.%
	As expected, the optimization of the IRS platform leads to lower NMSE values in comparison with the non-optimized IRS under the same LoS-to-NLoS SNR. This reveals both the potential of using the virtual link provided by IRS in  place of the LoS link when it is weak or obstructed and the gains provided by IRS optimization. Fig.~\ref{fig::6} shows the normalized estimation error of the back-scattering coefficient $\alp$ with respect to  the noise variance. It follows from  Fig.~\ref{fig::6} that when the LoS-to-NLoS SNR  is  set to $10^{-2}$, the NLoS outperforms the LoS link.
	%---------------------------------
	\begin{figure}[tb]
		\includegraphics[width=0.87\columnwidth]{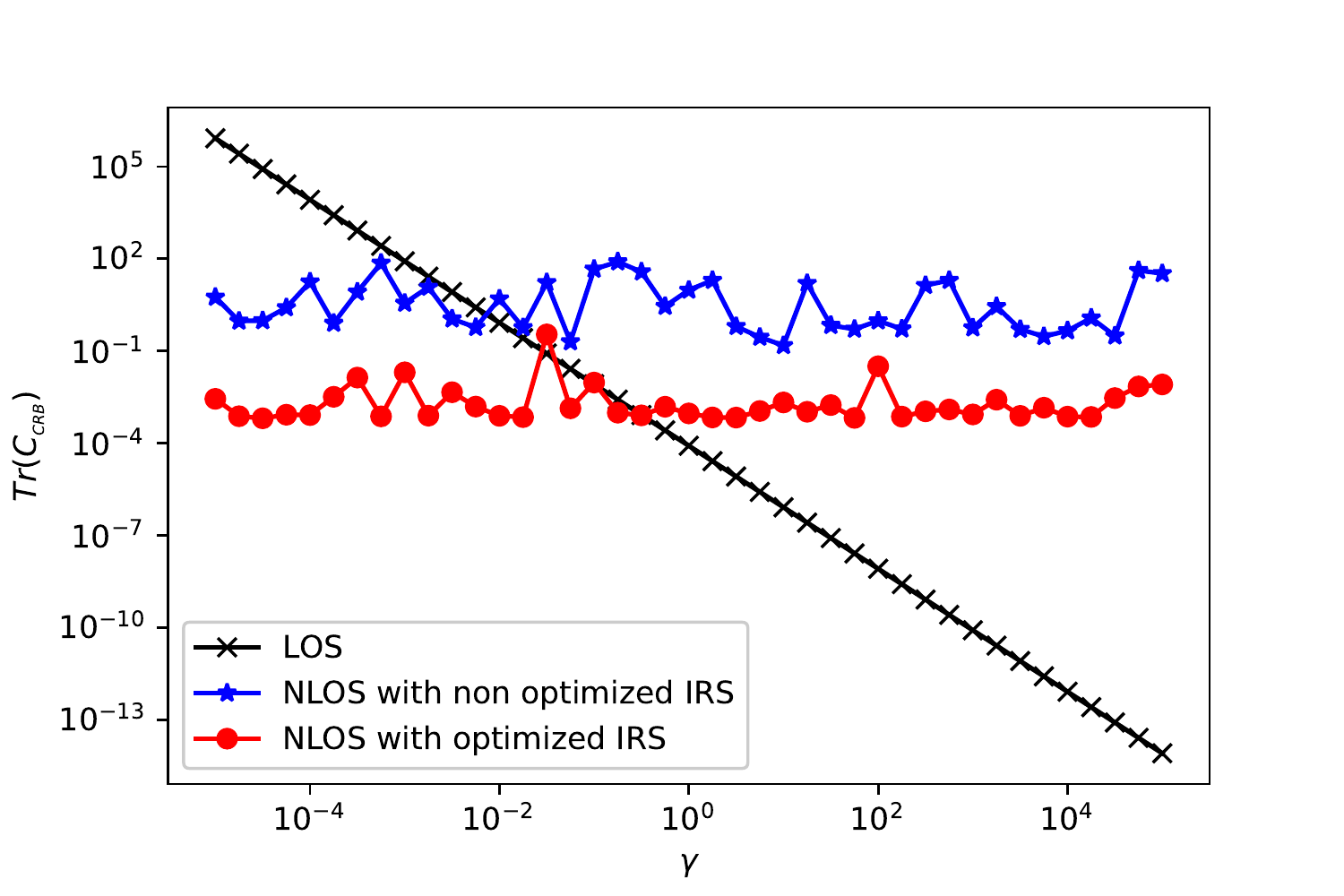}
		\caption{The CRB of the target scattering coefficient $\Tilde{\alp}$ for different values of  $\gamma\in [10^{-5},10^{5}]$, with $K=5$ and $M=10$.}
		\label{fig::9}
	\end{figure}
	%---------------------------------
	Fig.~\ref{fig::9} illustrates that the  CRB of the estimator $\hat{\alp}$ in IRS-aided radar overcomes the LoS links as weak as $10^{-1}$, i.e. in the same regime of $\gamma$, where IRS was effective as per the NMSE measure in Fig.~\ref{fig::2}. For illustration, the  $A$-optimality criteria i.e.   $\textrm{Tr}(\C_{_{\textrm{CRB}}})$ is chosen as a scalar measure of the CRB~\cite{tohidi2018sparse}.
	\section{Summary}
	\label{sec::conclusion}
	We studied the deployment of IRS in narrowband radar sensing and we presented an initial study on the effectiveness of IRS in assisting target estimation in radar. The formulation proposed in this paper is useful as a baseline for other IRS-aided radar settings. We derived the optimal IRS phases in terms of the mean square error of target parameter estimation. %We evaluated the CRB for target parameter estimation in the optimized and non optimized IRS deployment. 
	We indicated that IRS aids in target parameter estimation when the LoS link is weaker in relative SNR by \textasciitilde$10^{-1}$ than the NLoS link. Our numerical experiments reveal the effectiveness of the IRS even with non-optimized phase shifts.
	
%	\clearpage
	\balance
\bibliographystyle{IEEEtran}
\bibliography{refs}
\end{document}